\documentclass[12pt]{article}

\usepackage{amsmath,amsfonts,amssymb,amsthm}
\usepackage{mathtools}
\usepackage{mathrsfs}
\usepackage{sgame}
\usepackage{float}
\usepackage{xcolor}
\usepackage{tikz}
\usetikzlibrary{automata, positioning, arrows, arrows.meta}
\usepackage{tikz-cd}
\usepackage{setspace}
\usepackage{geometry}
\usepackage{natbib}
\usepackage{authblk}
\usepackage[colorlinks,citecolor=blue,urlcolor=blue]{hyperref}
\usepackage[T1]{fontenc}
\newcommand\nodeab{\textcolor{red}{ab}}
\newcommand\nodeb{\textcolor{violet}{b}}
\newcommand\nodea{\textcolor{brown}{a}}

\numberwithin{equation}{section}
\theoremstyle{plain}
\newtheorem{theorem}{Theorem}[section]
\newtheorem{lemma}{Lemma}[section]

\newtheorem{claim}{Claim}[section]

\newtheorem{definition}{Definition}[section]
\theoremstyle{definition}

\newcommand{\ICR}{{\rm ICR}}
\newcommand{\BR}{{\rm br}}
\newcommand\naturals{\mathbb N}
\newcommand{\marg}{\text{marg}}
\newcommand{\proj}{\operatorname{proj}}
\newcommand{\id}{\text{id}}

\begin{document}\onehalfspacing
\title{Strategic Type Spaces\thanks{The authors are grateful to Stephen Morris, Muhamet Yildiz, Satoru Takahashi, Daisuke Oyama, Marcin P{\c e}ski, Michael Greinecker, and participants in the Transatlantic Theory Workshop, the Stony Brook Conference in Game Theory, the Institut Henri Poincar\'e Game Theory Seminar, the PSE TOM Seminar, the MIT Theory Lunch, and the One World Mathematical Game Theory Seminar for useful comments. Olivier Gossner acknowledges support from the French National Research Agency (ANR), ``Investissements d'Avenir'' (ANR-11-IDEX-0003/LabEx Ecodec/ANR-11-LABX-0047), UKRI, and the European Union (SInfoNiA ERC-2023-ADG 101142530).}}

\author{Olivier Gossner\thanks{CNRS--\'Ecole Polytechnique and London School of Economics} \quad Rafael Veiel\thanks{University of Texas at Austin}}

\maketitle

\begin{abstract}
We study how much information about a Harsanyi type is needed to determine its iterated best replies. For any finite incomplete-information game, we define Strategic Quotients, which retain only the finite-order descriptions needed for this computation. We show that there is a unique smallest such representation, the Strategic Type Space (STS). Its types are exactly interim correlated rationalizability hierarchies, and each hierarchy is rationalized by a belief over nature and opponents' strategic types. For each player, these hierarchies are the label sequences generated by a finite automaton.
\end{abstract}

\section{Introduction}

As shown by \cite{harsanyi1967games}, a player facing uncertainty about payoffs cares about the payoff-relevant state of nature, but also about the actions of her opponents. Since opponents' actions depend on their beliefs about nature, the player's beliefs about others' beliefs about nature are strategically relevant. In turn, opponents' choices are influenced by their beliefs about beliefs, so the player's beliefs about those higher-order beliefs are strategically relevant as well, and so on.

Harsanyi type spaces encode this recursive structure by associating each type with a belief over nature and the other players' types. The universal type space of \cite{mertens1985formulation} represents the resulting hierarchies of beliefs. Because it is game-independent, it retains every distinction that may matter in some game. This generality comes with a complexity cost. Even when nature has only two states, first-order beliefs form an interval such as $[0,1]$, second-order beliefs are probability distributions over nature and that interval, and higher-order beliefs iterate this construction. Eliciting higher-order beliefs and designing information on the universal type space are therefore difficult.

We instead fix a game and retain only distinctions that affect strategic behavior in that game. This permits simpler, game-dependent descriptions of information.

We call such a reduced representation a \emph{type space quotient}. A reduced type summarizes a class of beliefs, and every belief over nature and opponents' reduced types is assigned a reduced type. We focus on quotients that retain enough information to compute iterated best replies. At each finite order, they record the best replies to beliefs over opponents' preceding descriptions; consistency across orders gives a strategic-description tower. A \emph{Strategic Quotient} is a type space quotient admitting such a tower.

We define the Strategic Type Space (STS) as the unique smallest Strategic Quotient: every Strategic Quotient has a unique belief-preserving reduction to it. We prove that the STS exists and that its types are exactly ICR hierarchies. Each hierarchy is rationalized by a belief over nature and opponents' ICR hierarchies.

Our main result shows that, for every finite game, the STS is generated by a finite automaton. Universal belief hierarchies become more complex at each order, whereas the successive strategic descriptions of a fixed finite game obey a finite recursion.

Recursive arguments already underlie strategic contagion in binary-action supermodular games with action-dominant states, including the Email game of \cite{rubinstein1989electronic}, the robustness analysis of \cite{oyama2020generalized}, and the information-design analysis of \cite{morris22implementation}. Our construction applies to every finite game. The examples in Section \ref{sec:GlobalGames} illustrate how the recursion depends on the game's payoffs.

The remainder of the paper is organized as follows. Section~\ref{Litt} relates our construction to the literature on rationalizability and universal type spaces. Section~\ref{sec:notation} introduces the notation. Section~\ref{sec:ICR} defines interim correlated best replies, type space quotients, and strategic-description towers; it then constructs the STS as the space of ICR hierarchies and proves its existence and uniqueness. Section~\ref{sec:FiniteSTS} proves that the continuation state space is finite, derives the finite-automaton representation, and illustrates it in two binary-action examples. Section~\ref{sec:Dis} discusses other solution concepts, common-prior models, and possible extensions. The appendix contains the proofs omitted from the main text, including the finite-state argument.

\section{Related Literature}\label{Litt}

\paragraph{Rationalizability and Information} Rationalizability was introduced by \cite{bernheim1984rationalizable, pearce1984rationalizable} for complete-information games. \cite{dekel2007interim} extend it to incomplete-information games through interim correlated rationalizability (ICR). For each type, ICR repeatedly removes actions that are not best replies under any state-contingent, correlated belief over opponents' types and actions compatible with that type's belief. \cite{dekel2007interim} show that two Harsanyi types have the same ICR actions in every game if and only if they induce the same hierarchy of beliefs. They also show that the actions surviving round $m$ are measurable with respect to the type's order-$m$ beliefs. This result motivates our coarser, game-dependent representation.

Several papers relate information coarsenings to ICR. \cite{chen2016characterizing} define frames, which are finite partitions of belief hierarchies that respect the types' belief structure and are therefore special cases of type space quotients. \cite{chen2014robust} study collections in the universal type space that are closed under best replies. \cite{chenrefine} provide an algorithm for computing ICR hierarchies and use it to study refinements. We instead fix a game and construct its smallest representation of strategically relevant information without starting from a particular Harsanyi type space. Our finite-state recursion for ICR is new.

\paragraph{Universal Type Spaces} \cite{mertens1985formulation} construct the universal type space recursively from finite-order beliefs, using topology to extend complete belief hierarchies. \cite{ely2004hierarchies} give a related strategic characterization for a version of interim rationalizability that restricts correlation.\footnote{See Section \ref{sec:Dis}.} Both constructions concern information that may matter in any game.

Our construction, like the purely measurable universal space of \cite{HeifetzSamet1998}, imposes no topological assumptions on the input Harsanyi type spaces. General measurable state spaces raise extension questions that depend on the chosen definition of a belief hierarchy.\footnote{\cite{Pinter2012} defines a particular purely measurable hierarchy space and asserts the required inverse-limit property for every hierarchy in that space and provides its representation in the complete universal type space but distinguishes this hierarchy construction from the one used in the Heifetz--Samet counterexample.}  The simplification we obtain by fixing the game is a countable space of strategic hierarchies with finitely many continuation states.

Every finite best-reply hierarchy extends to a complete hierarchy, and every complete hierarchy is rationalized by at least one belief over nature and strategic types.

\section{Preliminaries and Notation}\label{sec:notation}
All spaces are measurable, and products carry the product $\sigma$-algebra. For a family $(X_i)_i$, write $X\coloneqq\prod_iX_i$ and $X_{-i}\coloneqq\prod_{j\neq i}X_j$. For maps $f_i\colon X_i\to Y_i$, write $f=(f_i)_i$ and $f_{-i}=(f_j)_{j\neq i}$ for the corresponding product maps. We write $(\id_K,f_{-i})$ for the map $(k,x_{-i})\mapsto(k,f_{-i}(x_{-i}))$, where $\id_K$ is the identity on $K$.

For a sequence $x_i=(x_{i,n})_n$, let $\proj_n(x_i)=x_{i,n}$, and let $\proj_{-i,n}$ be the corresponding product map, so $\proj_{-i,n}(x_{-i})=(x_{j,n})_{j\neq i}$. The cardinality of the underlying set of $X$ is denoted by $|X|$.
 
The set of probability measures on $X$ is denoted by $\Delta_X$ and carries the cylinder $\sigma$-algebra.\footnote{This is the $\sigma$-algebra generated by the sets $\{p\in\Delta_X:p(E)\in[a,b]\}$ for measurable $E\subseteq X$ and $0\leq a\leq b\leq1$.} For $Z=Y\times\prod_{m\in M}X_m$, where $M$ is at most countable, $\marg_{Y,m}(p)$ denotes the marginal of $p\in\Delta_Z$ on $Y\times X_m$. For a measurable map $f\colon X\to Y$, $p\circ f^{-1}$ denotes the distribution induced by $f$. For countable spaces we omit measurability qualifications, and $\operatorname{supp}p$ denotes the support of a countably supported distribution.

In diagrams, an arrow from $X$ to $Y$ on the subscripts denotes the induced map from $\Delta_X$ to $\Delta_Y$:
\begin{center}
\begin{tikzcd} 
 \Delta_{X}     \arrow[shift left=1ex]{d}[swap]{}     \\
\Delta_{Y}
 \end{tikzcd}
\end{center} 

Double-headed arrows denote onto maps. The finite set of players is $I$, the finite state space is $K$, the finite action sets are $(A_i)_{i\in I}$, and payoffs are given by $u\colon A\times K\to\mathbb R^I$.

\section{Strategic Type Spaces} \label{sec:ICR}

This section introduces strategic type spaces. Subsection \ref{sec:ICR_def} defines the best-reply correspondence underlying interim correlated rationalizability. Subsection \ref{sec:Axioms} defines Strategic Quotients and the smallest such quotient, the STS. Subsection \ref{sec:Charact} constructs the STS as the space of best-reply hierarchies and establishes existence and uniqueness. Section \ref{sec:FiniteSTS} gives its finite-automaton representation.

\subsection{Interim Correlated Best Replies}\label{sec:ICR_def}

A \emph{Harsanyi type space} $\mathscr{H}$ consists of a family of measurable spaces $(T_i)_{i\in I}$ and measurable maps $\pi_i\colon T_i\to\Delta_{K\times T_{-i}}$, where $\pi_i(t_i)$ represents type $t_i$'s belief over the state and the other players' types.

Let $\mathcal B_i\coloneqq 2^{A_i}\setminus\{\varnothing\}$. A belief $p\in\Delta_{K\times\mathcal B_{-i}}$ specifies the state and an available action set for each opponent. A compatible joint distribution additionally selects an action from every announced set. Formally, let $\mathcal Q_i(p)$ be the set of \emph{compatible joint distributions} $q\in\Delta_{K\times\mathcal B_{-i}\times A_{-i}}$ satisfying
$$
\marg_{K,\mathcal B_{-i}}q=p
\quad\text{and}\quad
q\bigl(\{(k,b_{-i},a_{-i}):a_{-i}\in b_{-i}\}\bigr)=1.
$$
Player $i$'s best-reply correspondence $\BR_i\colon\Delta_{K\times\mathcal B_{-i}}\to\mathcal B_i$ is
$$
\BR_i(p)\coloneqq
\bigcup_{q\in\mathcal Q_i(p)}
\arg\max_{a_i\in A_i}\mathbb E_q\bigl[u_i(a_i,a_{-i},k)\bigr].
$$

As in \cite{dekel2007interim}, ICR on a Harsanyi type space is defined recursively by
$$
\ICR_i^0(t_i)\coloneqq A_i,
\qquad
\ICR_i^m(t_i)\coloneqq
\BR_i\left(\pi_i(t_i)\circ(\id_K,\ICR_{-i}^{m-1})^{-1}\right),
\quad m\geq1.
$$
The sequence $(\ICR_i^m(t_i))_{m\geq0}$ is the \emph{\ICR-hierarchy} of $t_i$, and $\ICR_i(t_i)\coloneqq\bigcap_m\ICR_i^m(t_i)$ is its ICR set.

We now state the monotonicity property of $\BR$ used throughout the paper. Beliefs $p_i^-,p_i^+\in\Delta_{K\times\mathcal B_{-i}}$ admit a \emph{nested coupling} if there is a distribution $\lambda_i\in\Delta_{K\times\mathcal B_{-i}\times\mathcal B_{-i}}$ whose marginals on $(k,b_{-i}^-)$ and $(k,b_{-i}^+)$ are respectively $p_i^-$ and $p_i^+$ and such that
$$
\lambda_i\bigl(\{(k,b_{-i}^-,b_{-i}^+):b_j^-\subseteq b_j^+
\text{ for every }j\neq i\}\bigr)=1.
$$

\begin{lemma}[Monotonicity of $\BR$]\label{MonoBR}
If $p_i^-$ and $p_i^+$ admit a nested coupling, then
$$
\BR_i(p_i^-)\subseteq\BR_i(p_i^+).
$$
\end{lemma}

The proof is in Appendix \ref{Prelim}.

\subsection{Strategic Type Spaces}\label{sec:Axioms}

We first define type space quotients. A type space quotient assigns a reduced type to every belief over the state and opponents' reduced types, while preserving the belief structure of every Harsanyi type space.

\begin{definition}[Type Space Quotient]\label{TS}
A pair $\mathscr{S} = (\mathcal{S}_i,\beta_i)_{i}$ is a \emph{Type Space Quotient} if the maps $\beta_i$ are measurable surjections and,
for every Harsanyi type space $\mathscr{H} = (T_i,\pi_i)_{i}$, there exists a family of measurable maps $(\chi_i)_{i}$ for which the following diagram commutes:
\begin{center}
\begin{tikzcd} 
T_i\arrow{d}{\chi_i} \arrow{r}{\pi_i} & \Delta_{K \times T_{-i}}  \arrow[shift left=1.75ex]{d}{\chi_{-i}}    \arrow[shift right=1.2ex]{d}[swap]{id}     \\
\mathcal{S}_i &  \arrow{l}{\beta_i} \Delta_{K \times \mathcal{S}_{-i}}     
 \end{tikzcd}
\end{center} 
\end{definition}

The diagram says that a type's reduced description depends only on its induced belief over $K\times\mathcal S_{-i}$. Thus two Harsanyi types with the same induced belief receive the same reduced description. The converse need not hold: distinct induced beliefs may receive the same description.

Reduced types therefore partition beliefs.\footnote{\cite{chen2016characterizing} define frames as partitions of belief hierarchies in the universal type space of \cite{mertens1985formulation} that respect the types' belief structure. Frames are special cases of type space quotients.} In a Harsanyi type space, each type has a unique belief. A type space quotient instead allows several beliefs to share the same reduced type and can therefore discard distinctions present in the original type space.

The singleton $\mathcal S_i=\{*\}$ with the constant map $\beta_i\colon\Delta_{K\times\mathcal S_{-i}}\to\{*\}$ is a type space quotient for any game. Thus the quotient property alone may discard strategically relevant information. The Mertens--Zamir construction represents a type by a consistent hierarchy of finite-order beliefs. For a fixed game, we retain at each order only the resulting set of best replies to beliefs over opponents' preceding descriptions. Consistency across orders gives the following tower.

\begin{definition}[Strategic-Description Tower]\label{BR-meas}
Let $(\mathcal{S},\beta)$ be a type space quotient. A \emph{strategic-description tower} is a family of measurable maps
$$
d_i^m:\mathcal{S}_i\longrightarrow\mathcal{B}_i^{m+1},
\qquad i\in I,\quad m\geq 0.
$$
Writing $b_i^m$ for the last coordinate of $d_i^m$, the tower satisfies $d_i^0(s_i)=(A_i)$ and, for every $p_i\in\Delta_{K\times\mathcal{S}_{-i}}$ and every $m\geq 0$,
$$
d_i^{m+1}\bigl(\beta_i(p_i)\bigr)
=
\left(
d_i^m\bigl(\beta_i(p_i)\bigr),
\BR_i\bigl(p_i\circ(\id_K,b_{-i}^m)^{-1}\bigr)
\right).
$$
\end{definition}

The tower preserves every finite round of \emph{synchronized} ICR: all opponents are evaluated at the same round $m$. It does not require a description to determine best replies when different opponents are evaluated at different rounds.

Because $\beta_i$ is onto, the order-$(m+1)$ description contains the order-$m$ description as its prefix. The tower records successive best-reply sets while discarding belief distinctions that never affect them. It parallels the Mertens--Zamir hierarchy: their recursion adds a belief over nature and opponents' preceding belief descriptions, whereas ours adds the resulting best-reply set.

\begin{definition}[Strategic Quotient]
A \emph{Strategic Quotient} is a triple $\mathscr S=(\mathcal S,\beta,d)$, where $(\mathcal S,\beta)$ is a type space quotient and $d$ is a strategic-description tower.
\end{definition}

The universal type space, endowed with its induced finite-order strategic descriptions, is a Strategic Quotient. In a game without uncertainty, the singleton type space is also a Strategic Quotient.

The next definition formalizes when one Strategic Quotient is a reduction of another. Let $\widetilde{\mathscr S}=(\widetilde{\mathcal S},\widetilde\beta,\widetilde d)$ and $\mathscr S=(\mathcal S,\beta,d)$ be Strategic Quotients.

\begin{definition}[Morphism]\label{reduction}
A \emph{morphism}, or reduction map, from $\widetilde{\mathscr S}$ to $\mathscr S$ is a family of measurable surjections $\zeta_i:\widetilde{\mathcal S}_i\to\mathcal S_i$ such that the following diagram commutes:
\begin{center}
\begin{tikzcd} 
\tilde{\mathcal{S}}_i \arrow[twoheadrightarrow]{d}{ \zeta_i} & \arrow{l}[swap]{\tilde{\beta}_i} \Delta_{K \times \tilde{\mathcal{S}}_{-i}}  \arrow[twoheadrightarrow, shift left=1.75ex]{d}{\zeta_{-i}}    \arrow[shift right=1.2ex]{d}[swap]{id}     \\
\mathcal{S}_i &  \arrow{l}{\beta_i} \Delta_{K \times \mathcal{S}_{-i}}     
 \end{tikzcd}
\end{center} 
and
$$
d_i^m\circ\zeta_i=\widetilde d_i^m
\qquad\text{for every }i\in I\text{ and }m\geq 0.
$$
The Strategic Quotient $\mathscr S$ is \emph{smaller than} $\widetilde{\mathscr S}$ if there exists a morphism from $\widetilde{\mathscr S}$ to $\mathscr S$.
\end{definition}

The first condition preserves beliefs over the reduced space, and the second preserves every finite-order strategic description. Thus a smaller quotient removes distinctions without changing either the induced beliefs or the strategic descriptions.

\begin{definition}[STS]\label{Mini}
A \emph{strategic type space} (STS) is a terminal Strategic Quotient: every Strategic Quotient admits a unique morphism to it.
\end{definition}

Equivalently, the STS is the unique smallest Strategic Quotient. It retains one point for each complete strategic-description tower and no redundant copies of the same tower.

\subsection{The Strategic Type Space}\label{sec:Charact}

We now establish existence and uniqueness. First, the tower of any Strategic Quotient recovers the ICR hierarchy of every Harsanyi type (Lemma \ref{StratThm1}). We then construct the set $\mathcal S_i$ of all best-reply hierarchies directly and show that these are exactly the ICR hierarchies generated by types in Harsanyi type spaces (Lemma \ref{thm:ICR_hierarchies}). Finally, we define the belief maps $\beta_i$ and show that the resulting Strategic Quotient is the smallest one (Lemma \ref{RepThm1a} and Theorem \ref{CoroEx}).

Our first result states that every Strategic Quotient allows us to recover the ICR hierarchies from any Harsanyi type.

\begin{lemma}[Strategic Towers Recover ICR]\label{StratThm1}
Let $\mathscr S=(\mathcal S,\beta,d)$ be a Strategic Quotient. For every Harsanyi type space $(T_i,\pi_i)_i$ and associated maps $(\chi_i)_i$ satisfying Definition \ref{TS},
$$
d_i^m\bigl(\chi_i(t_i)\bigr)
=
\bigl(\ICR_i^0(t_i),\ldots,\ICR_i^m(t_i)\bigr)
$$
for every $t_i\in T_i$, $i\in I$, and $m\geq 0$.
\end{lemma}

Proofs not given in the main text are in the appendix.

We now construct the set $\mathcal S$ of all best-reply hierarchies. Level zero is the full action set. At each subsequent level, a player holds a belief over nature and opponents' preceding descriptions and records the resulting best-reply set. A finite hierarchy is admissible when one belief generates all its levels.

We define the sets $\mathcal S_i^m$ inductively. An order-$m$ hierarchy is an $(m+1)$-tuple $s_i^m=(A_i,b_i^1,\ldots,b_i^m)$, whose level-zero coordinate is $A_i$. Set $\mathcal S_i^0=\{A_i\}$. Given $\mathcal S^{m-1}$, let $\mathcal S_i^m$ consist of the tuples in $\mathcal S_i^{m-1}\times\mathcal B_i$ for which there exists $p_i\in\Delta_{K\times\mathcal S_{-i}^{m-1}}$ satisfying

\begin{equation}\label{BRhierarchylevelm}
\BR_i(\marg_{K,\ell}(p_i))=b_i^{\ell+1},
\qquad \ell=0,\ldots,m-1,
\end{equation}
where
$$
\marg_{K,\ell}(p_i)
\coloneqq
p_i\circ(\id_K,\proj_{-i,\ell})^{-1}
$$
is the marginal on the state and opponents' level-$\ell$ action sets. We define the set of player $i$'s best-reply hierarchies as
$$
\mathcal{S}_i
\coloneqq
\{s_i\in\mathcal{B}_i^{\mathbb N}:s_i^m\in\mathcal S_i^m
\text{ for every }m\in\mathbb N\}.
$$
We equip $\mathcal S_i$ with the $\sigma$-algebra generated by its coordinate projections. This is the full power set: by Claim \ref{Monoghier}, every hierarchy is decreasing and hence changes only finitely many times, so $\mathcal S_i$ is countable; the coordinates separate its points.

The next lemma shows that $\mathcal S$ is exactly the set of ICR hierarchies that can arise in Harsanyi type spaces.

\begin{lemma}[Best-Reply Hierarchies are ICR Hierarchies]\label{thm:ICR_hierarchies}
\begin{itemize}
\item[(i)] A tuple $s^m\in\mathcal B^{m+1}$ belongs to $\mathcal S^m$ if and only if there exist a Harsanyi type space $(T,\pi)$ and a type profile $t\in T$ such that $s^m=(\ICR^\ell(t))_{\ell\leq m}$.
\item[(ii)] A sequence $s\in\mathcal B^\naturals$ belongs to $\mathcal S$ if and only if there exist a Harsanyi type space $(T,\pi)$ and a type profile $t\in T$ such that $s=(\ICR^\ell(t))_{\ell\geq0}$.
\end{itemize}
\end{lemma}

The ``if'' direction in both parts follows from the definitions. The converse requires a single belief to generate every level of a complete hierarchy. Let $\beta_i^0$ be the constant map with value $(A_i)$. For every $m\geq1$, define the best-reply map on finite hierarchies $\beta_i^m:\Delta_{K\times\mathcal S_{-i}^{m-1}}\to\mathcal S_i^m$ by
\begin{equation}\label{Eq_BRSTS}
\beta_i^m(p_i) \coloneqq (A_i,\BR_i(\marg_{K,0}(p_i)), \dots, \BR_i(\marg_{K,m-1}(p_i))).
\end{equation}
A belief $p_i$ on $K \times\mathcal S_{-i}$ induces, for every $m\geq1$, a belief $p_i^m$ on $K\times\mathcal S_{-i}^{m-1}$ by projection on the first $m$ coordinates. The finite descriptions $\beta_i^m(p_i^m)$ agree under truncation and therefore define a unique hierarchy, denoted $\beta_i(p_i)$. This gives a map $\beta_i \colon \Delta_{K \times\mathcal{S}_{-i}} \to \mathcal{S}_i$. For $s_i=(s_{i,0},s_{i,1},\ldots)\in\mathcal S_i$, define its order-$m$ strategic description by
$$
d_i^m(s_i)\coloneqq(s_{i,0},\ldots,s_{i,m}),
$$
and let $b_i^m(s_i)\coloneqq s_{i,m}$. The proof shows that $\beta_i$ is onto. The difficulty is to find one belief that generates every round. Starting from a sufficiently long finite hierarchy, the proof replaces opponents' continuations so that they stop changing after a uniformly bounded number of rounds; the resulting belief then generates the complete hierarchy.

The next lemma shows that these hierarchy spaces form a Strategic Quotient.
 
 \begin{lemma}[ICR Hierarchies Form a Strategic Quotient]\label{RepThm1a}
 $(\mathcal S,\beta,d)$ is a Strategic Quotient.
 \end{lemma}

By Lemma \ref{StratThm1}, the tower of any Strategic Quotient recovers every finite-order ICR hierarchy. Lemma \ref{thm:ICR_hierarchies} identifies $\mathcal S$ with all complete ICR hierarchies, and Lemma \ref{RepThm1a} shows that this set is a Strategic Quotient. It remains to prove that this quotient is terminal.

\begin{theorem}[Existence and Uniqueness of the STS]\label{CoroEx}
The Strategic Quotient $(\mathcal S,\beta,d)$ is terminal. Consequently, the STS exists and is unique up to a unique isomorphism.
\end{theorem}

\section{Finite Representation of the STS}\label{sec:FiniteSTS}

After a finite hierarchy, several infinite tails may remain possible. We use the set of these possible tails as the state of the automaton.

For every $m\in\naturals$, player $i$, and finite hierarchy $s_i^m\in\mathcal S_i^m$, define its continuation set by
\begin{equation}\label{Om}
\mathcal T_i(s_i^m)
\coloneqq
\left\{
(\hat s_{i,m},\hat s_{i,m+1},\ldots):
\hat s_i\in\mathcal S_i
\text{ and }
\hat s_i^m=s_i^m
\right\}.
\end{equation}
Lemma \ref{thm:ICR_hierarchies} implies that every finite hierarchy has a complete extension, so these sets are nonempty. Let
\begin{equation}\label{O}
\Omega_i
\coloneqq
\left\{
\mathcal T_i(s_i^m):
s_i^m\in\mathcal S_i^m,\ m\in\naturals
\right\}.
\end{equation}

\begin{theorem}[Finite continuation state space]\label{FiniteO}
For every player $i$, the set $\Omega_i$ is finite.
\end{theorem}

The proof is in Appendix \ref{FiniteSTSProof}. Every ICR hierarchy is decreasing and therefore has at most $|A_i|-1$ strict changes. Only the lengths of the intervening constant periods can be unbounded. The finite best-reply update distinguishes only finitely many effects of these lengths. A finite summary of the successive action sets and the lengths of their constant periods therefore determines all possible continuations. Since there are finitely many such summaries, $\Omega_i$ is finite.

\subsection{Automaton Representation}

An automaton for player $i$ is a tuple
$\hat{\mathscr A}_i
=(\hat\Omega_i,\hat\ell_i,\widehat{\preceq}_i,\hat\omega_i^0)$,
where $\hat\Omega_i$ is a finite set of states,
$\hat\ell_i:\hat\Omega_i\to\mathcal B_i$ assigns a label to every state,
$\widehat{\preceq}_i$ is a successor relation, and
$\hat\omega_i^0\in\hat\Omega_i$ is the initial state. A path is a sequence
$(\omega_i^0,\omega_i^1,\ldots)$ such that
$\omega_i^0=\hat\omega_i^0$ and
$\omega_i^m\widehat{\preceq}_i\omega_i^{m+1}$ for every $m$.
Let $P_{\hat{\mathscr A}_i}$ denote the set of paths.

Every $\omega_i\in\Omega_i$ is nonempty and all its elements have the same first coordinate. Let $\ell_i(\omega_i)$ denote this coordinate. Define
$$
\gamma_i(b_i^0,b_i^1,b_i^2,\ldots)
\coloneqq
(b_i^1,b_i^2,\ldots)
$$
and, for $\omega_i\subseteq\mathcal B_i^\naturals$, let
$\gamma_i(\omega_i)\coloneqq\{\gamma_i(x_i):x_i\in\omega_i\}$.
For $\omega_i,\hat\omega_i\in\Omega_i$, set
\begin{equation}\label{AutomatonSuccessor}
\omega_i\preceq_i\hat\omega_i
\quad\Longleftrightarrow\quad
\hat\omega_i\subseteq\gamma_i(\omega_i).
\end{equation}
Finally, let
$$
\omega_i^0\coloneqq\mathcal T_i((A_i)).
$$
Since every hierarchy begins with $A_i$, $\omega_i^0=\mathcal S_i$.

\begin{theorem}[Finite automaton representation]\label{FiniteAutomaton}
For every player $i$,
$\mathscr A_i=(\Omega_i,\ell_i,\preceq_i,\omega_i^0)$ is a finite automaton satisfying
$$
\left\{
\bigl(\ell_i(\omega_i^0),\ell_i(\omega_i^1),\ldots\bigr):
(\omega_i^0,\omega_i^1,\ldots)\in P_{\mathscr A_i}
\right\}
=
\mathcal S_i.
$$
\end{theorem}

\begin{proof}
If $s_i\in\mathcal S_i$, then
$\omega_i^m\coloneqq\mathcal T_i(s_i^m)$ defines a path and
$\ell_i(\omega_i^m)=s_{i,m}$. Indeed, every complete extension of
$s_i^{m+1}$ is also an extension of $s_i^m$, so
$$
\mathcal T_i(s_i^{m+1})
\subseteq
\gamma_i\bigl(\mathcal T_i(s_i^m)\bigr).
$$

Conversely, consider a path $(\omega_i^m)_m$. The order-zero prefix is automatic. Fix $M\geq1$ and choose $x_i^M\in\omega_i^M$. Since
$\omega_i^{m+1}\subseteq\gamma_i(\omega_i^m)$, backward induction gives
$x_i^m\in\omega_i^m$ for $m=M-1,\ldots,0$ so that
$\gamma_i(x_i^m)=x_i^{m+1}$. This gives
$x_i^0\in\omega_i^0=\mathcal S_i$ whose first $M+1$ coordinates are
$\ell_i(\omega_i^0),\ldots,\ell_i(\omega_i^M)$.
Thus the first $M+1$ labels belong to $\mathcal S_i^M$ for every $M$. By the definition of $\mathcal S_i$, the complete label sequence belongs to $\mathcal S_i$.
\end{proof}

\subsection{Examples}\label{sec:GlobalGames}

We present two examples with two players, states $k_G,k_B$, and action sets $\{a,b\}$. When action sets appear as hierarchy labels, we write $a$, $b$, and $ab$ for $\{a\}$, $\{b\}$, and $\{a,b\}$. In state $k_G$, payoffs are given by the coordination game
\begin{center}\hfill
\begin{game}{2}{2}[$k_G$]
 &  $a$ & $b$ \\
$a$ & $ 1,1$ & $0,0$ \\
$b$ & $0,0$ & $1,1$
\end{game}\hfill~
\end{center}
In state $k_B$, each player has a dominant action. The players have the same dominant action in the first example and different dominant actions in the second. Dominance states initiate the propagation of strategic reasoning in information design \citep{morris22implementation} and global games \citep{carlsson1993global}. The examples show how the payoffs determine this propagation and hence the resulting STS.

\paragraph{Symmetric Dominance State}
In the first variant, action $b$ is strictly dominant for both players in state $k_B$:
\begin{center}\hfill
\begin{game}{2}{2}[$k_B$]
 &  $a$ & $b$ \\
$a$ & -1,-1 & $0,0$ \\
$b$ & $0,0$ & $1,1$
\end{game}\hfill~%
\end{center}%

To begin the hierarchy, consider beliefs $p_i$ assigning probability one to the opponent's full action set $\{a,b\}$. A compatible joint distribution preserves $p_i$'s marginal on the state and selects an opponent action in each state. Holding the state marginal fixed gives the rectangles shown in the right panel of Figure \ref{ICR1}. The left panel divides the simplex $\Delta_{\{k_G,k_B\}\times\{a,b\}}$ according to the resulting best replies. If $p_i(k_G)<1/2$, the entire rectangle lies in the region where $b$ is the unique best reply. If $p_i(k_G)>1/2$, the rectangle reaches regions where $a$ alone, $b$ alone, or both are best replies. At $p_i(k_G)=1/2$, action $a$ can tie with $b$ but cannot be the unique best reply.

\definecolor{colour0}{rgb}{0.8,0.8,0.8}
\definecolor{colour1}{rgb}{0.6,0.6,0.6}

\begin{figure}[H]
\begin{center}
\begin{tikzpicture}[xscale=0.55,yscale=0.5,every node/.style={inner sep=1pt}]
\draw[black,line width=0.02cm] (1.16,-1.815) -- (3.5364528,2.285) -- (6.6,-0.415) -- (5.027019,-2.495) -- cycle;
\draw[black,line width=0.02cm] (1.16,-1.795) -- (6.62,-0.415);
\draw[black,line width=0.02cm,dash pattern=on 0.17638889cm off 0.10583334cm,fill=colour0] (2.9,-2.075) -- (5.86,-1.375) -- (4.16,0.225) -- cycle;
\draw[black,line width=0.02cm] (3.52,2.265) -- (4.98,-2.435);
\node[anchor=south west] at (2.62,-2.595) {$1/2$};
\node[anchor=south west] at (3.56,0.065) {$1/2$};
\node[anchor=south west] at (5.94,-1.715) {$1/2$};
\node[anchor=south west] at (6.74,-0.535) {$(b,k_G)$};
\node[anchor=south west] at (3.04,2.345) {$(b,k_B)$};
\node[anchor=south west] at (4.68,-2.915) {$(a,k_G)$};
\node[anchor=south west] at (0.0,-2.3) {$(a,k_B)$};
\draw[black,line width=0.02cm] (8.94,-1.595) -- (11.316453,2.505) -- (14.38,-0.195) -- (12.807019,-2.275) -- cycle;
\draw[black,line width=0.02cm] (8.94,-1.575) -- (14.4,-0.195);
\draw[black,line width=0.02cm,dash pattern=on 0.17638889cm off 0.10583334cm,fill=colour0] (10.68,-1.855) -- (13.64,-1.155) -- (11.94,0.445) -- cycle;
\node[anchor=south west] at (10.4,-2.375) {$1/2$};
\node[anchor=south west] at (11.46,0.245) {$1/2$};
\node[anchor=south west] at (13.72,-1.495) {$1/2$};
\node[anchor=south west] at (14.52,-0.315) {$(b,k_G)$};
\node[anchor=south west] at (10.82,2.565) {$(b,k_B)$};
\node[anchor=south west] at (12.46,-2.7) {$(a,k_G)$};
\node[anchor=south west] at (7.78,-2.1) {$(a,k_B)$};
\draw[black,line width=0.02cm,fill=colour1] (9.76,-1.675) -- (10.3,-1.235) -- (12.0,1.785) -- (11.6,1.545) -- cycle;
\draw[black,line width=0.02cm,fill=colour1] (11.98,-2.055) -- (12.68,-1.355) -- (12.92,-0.575) -- (12.4,-1.115) -- cycle;
\draw[black,line width=0.02cm] (11.3,2.485) -- (12.76,-2.215);
\draw[black,line width=0.02cm,fill=colour1] (12.98,-0.495) -- (13.72,0.285) -- (13.54,-0.375) -- (13.2,-0.775) -- cycle;
\node[anchor=south west] at (14.02,0.365) {$p_1>1/2$};
\node[anchor=south west] at (12.34,1.685) {$p_1<1/2$};
\end{tikzpicture}
\end{center}
\caption{\small In the region between the shaded triangle and the triangle with vertices $(a,k_B)$, $(b,k_B)$, and $(b,k_G)$, $b$ is the unique best reply. In the region between the shaded triangle and $(a,k_G)$, $a$ is the unique best reply. On the shaded triangle, both actions are best replies. The boundary of the shaded triangle is excluded from the two outer regions.}\label{ICR1}
\end{figure}

For any player $i$ and any belief $p_i$ assigning probability one to the opponent's full action set,
$$
\BR_i(p_i)=b
\quad\Longleftrightarrow\quad
p_i(k_B)>p_i(k_G).
$$
Thus the first-round best-reply set is
\begin{itemize}
    \item[(1)] $b$ if $p_i(k_B)>p_i(k_G)$;
    \item[(2)] $ab$ if $p_i(k_B)\leq p_i(k_G)$.
\end{itemize}
In particular, no first-round belief yields $a$. At later rounds, beliefs are supported on $K\times\{ab,b\}$, and
\begin{itemize}
    \item[($i$.1)] $b$ if $p_i(k_G,b)+p_i(k_B)>p_i(k_G,ab)$;
    \item[($i$.2)] $ab$ otherwise.
\end{itemize}

Hence $ab$ may be followed by $ab$ or $b$, while $b$ is absorbing. Every path is realized by the standard email-chain construction with $T_i=\mathbb N\cup\{\infty\}$: type $0$ is certain of $(k_B,0)$, type $t>0$ is certain of $(k_G,t-1)$, and type $\infty$ is certain of $(k_G,\infty)$. The latter remains at $ab$ forever. Figure~\ref{Automa} therefore represents the STS exactly.

\begin{figure}[H]
\centering
\begin{tikzpicture}[
    ->,
    >=stealth,
    node distance=2.5cm,
    every state/.style={thick, fill=gray!10, minimum size=1cm},
    initial text=start,
    thick,
    shorten >=2pt,
    shorten <=2pt
]
    \node[blue,state,initial left,initial text={Start for Player $i$}] (q1) {$\nodeab$};

    \node[blue,state,accepting,above of=q1] (q3) {$\nodeb$};

    \draw
        (q1) edge[loop right,thick] node {\footnotesize (2), (1.2) or (2.2)} (q1)
(q1) edge[thick] node[right, xshift=2pt] {\footnotesize (1), (1.1) or (2.1)} (q3);
\end{tikzpicture}
\caption{STS automaton for player $i$ in the symmetric dominance-state example. The bottom node is initial, and the double-circled node is absorbing.}
\label{Automa}
\end{figure}

\paragraph{Asymmetric Dominance State}

In the second variant, $b$ is strictly dominant for player~1 and $a$ is strictly dominant for player~2 in state $k_B$:

\begin{center}\hfill
\begin{game}{2}{2}[$k_B$]
 &  $a$ & $b$ \\
$a$ & $0,1$ & $0,0$ \\
$b$ & $1,1$ & $1,0$
\end{game}\hfill~%
\end{center}%
For a first-round belief $p_i$, the best-reply sets are
\begin{itemize}
    \item[(1.1)] $b$ if $i=1$ and $p_i(k_B)>p_i(k_G)$;
    \item[(2.1)] $a$ if $i=2$ and $p_i(k_B)>p_i(k_G)$;
    \item[(3.1)] $ab$ if $p_i(k_B)\leq p_i(k_G)$.
\end{itemize}
For (1.1), the compatible selection least favorable to $b$ chooses $a$ at $(k_G,ab)$; for (2.1), the selection least favorable to $a$ chooses $b$ there. At later rounds, player~1's best-reply set is
\begin{itemize}
    \item[(1.2)] $a$ if $p_1(k_G,a)>p_1(k_G,b)+p_1(k_G,ab)+p_1(k_B)$;
    \item[(1.3)] $b$ if $p_1(k_G,a)+p_1(k_G,ab)<p_1(k_G,b)+p_1(k_B)$;
    \item[(1.4)] $ab$ otherwise.
\end{itemize}
For (1.2), the least favorable selection maps $(k_G,ab)$ to $b$ and $(k_B,ab)$ to $a$; for (1.3), it maps $(k_G,ab)$ to $a$ and $(k_B,ab)$ to $b$.

For player~2, the corresponding conditions are
\begin{itemize}
    \item[(2.2)] $a$ if $p_2(k_G,a)+p_2(k_B)>p_2(k_G,b)+p_2(k_G,ab)$;
    \item[(2.3)] $b$ if $p_2(k_G,a)+p_2(k_G,ab)+p_2(k_B)<p_2(k_G,b)$;
    \item[(2.4)] $ab$ otherwise.
\end{itemize}
The least favorable selections are the same as for player~1.

These inequalities first imply a restriction at round two. If player~1 retains $ab$ in round one, then $p_1(k_B)\leq p_1(k_G)$. Player~2's first-round set is never $b$, so the induced round-two belief satisfies $p_1(k,b)=0$ for every $k$. Condition (1.3) would require $p_1(k_G)<p_1(k_B)$, a contradiction. Thus player~1 cannot move from $ab$ to $b$ \emph{at round two}. Symmetrically, player~2 cannot move from $ab$ to $a$ at round two.

To establish the restrictions at every round, fix an arbitrary measurable Harsanyi type space. Write $p_i^B(t_i)=\pi_i(t_i)(\{k_B\}\times T_{-i})$ and, for $x\in\{a,b\}$, define
\[
g_{i,x}^r(t_i)=\pi_i(t_i)\{(k_G,t_{-i}):\ICR_{-i}^r(t_{-i})=\{x\}\}.
\]
The best-reply conditions above become, with type arguments suppressed,
\[
\begin{aligned}
\ICR_1^{r+1}=\{a\}&\quad\Longleftrightarrow\quad g_{1,a}^r>1/2,\\
\ICR_1^{r+1}=\{b\}&\quad\Longleftrightarrow\quad p_1^B+g_{1,b}^r>1/2,\\
\ICR_2^{r+1}=\{a\}&\quad\Longleftrightarrow\quad p_2^B+g_{2,a}^r>1/2,\\
\ICR_2^{r+1}=\{b\}&\quad\Longleftrightarrow\quad g_{2,b}^r>1/2.
\end{aligned}
\]
The remaining outcome is $ab$; equality at a threshold retains both actions. At round one, only $b$ can first become a singleton for player~1, and only $a$ for player~2. Since ICR hierarchies decrease, each singleton event grows with the round. For $r\geq1$, a threshold can newly be crossed at round $r+1$ only if the corresponding opponent singleton event has gained types at round $r$: otherwise its probability, and hence the comparison, is unchanged. Induction therefore implies that player~1 can first reach $b$ only at odd rounds and $a$ only at even rounds, with the opposite pattern for player~2. This proves necessity of the transitions in Figure~\ref{fig:automaton1ptech}; the type construction below proves sufficiency.

For each player, the designated $ab$ node is the round-zero state. A path alternates between the two $ab$ nodes until it reaches an absorbing singleton, if ever.

\begin{figure}[H]
\centering
\begin{tikzpicture}[
    ->,
    >=stealth,
    node distance=3.5cm,
    every state/.style={thick, fill=gray!10, minimum size=1cm},
    initial text=start,
    thick,
    shorten >=2pt,
    shorten <=2pt
]
    \node[blue,state,initial left,initial text=start for P1] (q1) {$\nodeab$};
    \node[blue,state,right of=q1,initial right,initial text=start for P2] (q2) {$\nodeab$};
    \node[blue,state,accepting,above of=q1] (q3) {$\nodeb$};
    \node[blue,state,accepting,above of=q2] (q4) {$\nodea$};

    \draw
        (q1) edge[bend right,thick]
            node[below, yshift=-2pt] {\footnotesize (3.1),(1.4) or (2.4)}
            (q2)

        (q2) edge[bend right,thick]
            node[above, yshift=2pt] {\footnotesize (3.1),(1.4) or (2.4)}
            (q1)

        (q1) edge[thick]
            node[left, xshift=-2pt] {\footnotesize (1.1), (1.3) or (2.3)}
            (q3)

        (q2) edge[thick]
            node[right, xshift=2pt] {\footnotesize (2.1), (1.2) or (2.2)}
            (q4);
\end{tikzpicture}
\caption{STS automata in the asymmetric dominance-state example. The left $ab$ node is initial for player~1, the right $ab$ node is initial for player~2, and the double-circled nodes are absorbing.}
\label{fig:automaton1ptech}
\end{figure}

It remains to verify that every path is feasible. An information structure similar to that of \cite{rubinstein1989electronic} does so. Let each player's type set be
$T_i=\mathbb N\cup\{\infty\}$. For finite types, define
$$
\pi_i(k,t_{-i}\mid t_i)
\coloneqq
\begin{cases}
\frac13,
&t_i>0\text{ and }(k,t_{-i})=(k_G,t_i+1),\\
\frac23,
&t_i>0\text{ and }(k,t_{-i})=(k_G,t_i-1),\\
1,
&t_i=0\text{ and }(k,t_{-i})=(k_B,0),\\
0,
&\text{otherwise},
\end{cases}
$$
and let type $\infty$ assign probability one to $(k_G,\infty)$. Induction using the preceding best-reply conditions gives, for every $m\in\mathbb N$,
$$
\ICR_1^m(t_1)
=
\begin{cases}
b,
&t_1\in2\mathbb N\text{ and }m>t_1,\\
a,
&t_1\in2\mathbb N+1\text{ and }m>t_1,\\
ab,
&\text{otherwise},
\end{cases}
$$
and
$$
\ICR_2^m(t_2)
=
\begin{cases}
a,
&t_2\in2\mathbb N\text{ and }m>t_2,\\
b,
&t_2\in2\mathbb N+1\text{ and }m>t_2,\\
ab,
&\text{otherwise}.
\end{cases}
$$
A finite type $t_i$ reaches an absorbing singleton at round $t_i+1$, with the singleton determined by parity. Type $\infty$ remains at $ab$ forever and therefore generates the path that cycles forever between the two nonabsorbing nodes. Hence Figure~\ref{fig:automaton1ptech} represents the STS exactly.

The proof of Theorem~\ref{FiniteO} tracks this propagation through a finite update rule. Here, propagation alternates across players and produces the period two visible in Figure~\ref{fig:automaton1ptech}. In a general finite game, the initial phase and the eventual cycle may be longer.

\section{Discussion}\label{sec:Dis}

\paragraph{Solution concepts and common priors}
We study information through its implications for ICR. This solution concept is defined directly on Harsanyi type spaces, can be computed separately for each type, imposes no common-prior restriction, and admits the recursion used throughout the paper.

We now briefly discuss alternative approaches, which rely on different solution concepts. 

Related concepts include interim independent rationalizability (IIR) \citep{ely2004hierarchies} and interim partially correlated rationalizability (IPCR) \citep{tang2015interim}. For these concepts, action-set descriptions may not record enough information about how a player correlates opponents' actions. An analogous construction would therefore have to enrich each level with finite information about the allowed correlations. Whether finitely many such descriptions suffice for a fixed game remains open. For $\Delta$-rationalizability \citep{battigalli2003rationalization}, the tower would also depend on the imposed belief restriction $\Delta$.

In a companion paper \citep{GossnerVeiel24}, we represent ICR hierarchies in common-prior models as Markov chains on finite state spaces and use this representation to characterize rationalizable outcomes. Under a common prior, one may instead study belief-invariant Bayes correlated equilibrium (BIBCE) as in \cite{liu2015correlation}.

Bayes--Nash equilibrium, belief-invariant Bayesian solutions, and agent-normal-form correlated equilibrium \citep{forges1993five} raise a related problem: their outcomes may depend on correlation possibilities not recorded by action-set hierarchies. Fixing the game may again permit a coarser representation, but these concepts lack the iterative procedure used in our construction.

\paragraph{Universal Type Spaces as Strategic Quotients}

The universal spaces discussed in Section \ref{Litt} can also be viewed as Strategic Quotients. In the Mertens--Zamir space, a universal type and its belief over nature and opponents' universal types determine one another. The inverse belief map gives the quotient structure, and the measurability result of \cite{dekel2007interim} supplies the ICR tower for every finite game.

\cite{ely2004hierarchies} instead retain exactly the information needed for interim independent rationalizability in every finite game. Their construction records hierarchies of beliefs conditional on $K$ and relies on topology. It carries the corresponding strategic descriptions for IIR, but it does not satisfy our quotient requirement on all measurable Harsanyi type spaces.

The purely measurable universal space of \cite{HeifetzSamet1998} applies to the same class of Harsanyi type spaces as our construction. Its belief map is onto and has a measurable inverse, so this space is a type space quotient; \cite{dekel2007interim} again supplies its ICR tower. After fixing the game, hierarchies of best-reply sets replace full hierarchies of beliefs. This gives a countable strategic representation with finitely many continuation states. No topology is imposed on the input type spaces; the proof does use compactness on an auxiliary finite-alphabet sequence space.

\paragraph{Further Extensions} The STS may extend to infinite games through the version of ICR in \cite{weinstein2011interim}. The automaton representation relies on finite state and action sets and would not extend directly. For continuous actions, endpoint descriptions need not give a finite state space.

\bibliographystyle{ecta}
\bibliography{MyBibFile}

\appendix

\section{Proofs}

\subsection{Preliminaries}\label{Prelim}

We first show $\BR_i$ is measurable. For each $a_i\in A_i$, let
\[
C_{i,a_i}=\{p\in\Delta_{K\times\mathcal B_{-i}}:a_i\in\BR_i(p)\}.
\]
This set is the projection on $p$ of the compact set of pairs $(p,q)$ satisfying $q\in\mathcal Q_i(p)$ and the finitely many payoff inequalities making $a_i$ optimal. Thus $C_{i,a_i}$ is closed. For $b_i\in\mathcal B_i$,
\[
\{p:\BR_i(p)=b_i\}
=\bigcap_{a_i\in b_i}C_{i,a_i}\;\cap\!
\bigcap_{a_i\notin b_i}C_{i,a_i}^{\,c}
\]
is Borel, although it need not be closed. Consequently, the following recursion defines measurable action-set maps on every measurable Harsanyi type space.

We begin with the monotonicity lemma used throughout the paper.

\paragraph{Lemma \ref{MonoBR}} \emph{If $p_i^-$ and $p_i^+$ admit a nested coupling, then $\BR_i(p_i^-)\subseteq\BR_i(p_i^+)$.}

\begin{proof}
Let $\lambda_i$ be a nested coupling of $p_i^-$ and $p_i^+$, and fix $q_i^-\in\mathcal Q_i(p_i^-)$. Draw $(k,b_{-i}^-,b_{-i}^+)$ according to $\lambda_i$ and then draw $a_{-i}$ using the conditional probabilities in $q_i^-$ given $(k,b_{-i}^-)$. This is possible because $\lambda_i$ and $q_i^-$ have the same marginal $p_i^-$ on $(k,b_{-i}^-)$. Let $q_i^+$ be the resulting marginal on $(k,b_{-i}^+,a_{-i})$.

The marginal of $q_i^+$ on $(k,b_{-i}^+)$ is $p_i^+$. Moreover,
$a_{-i}\in b_{-i}^-\subseteq b_{-i}^+$ almost surely, so $q_i^+\in\mathcal Q_i(p_i^+)$. The distributions $q_i^-$ and $q_i^+$ have the same marginal on $K\times A_{-i}$ and therefore give every action of player $i$ the same expected payoff. Hence every action that is optimal under $q_i^-$ is also optimal under $q_i^+$.
\end{proof}

\subsection{Strategic Type Spaces}\label{STS}

\subsubsection{Characterization}
\paragraph{Lemma \ref{StratThm1}} \emph{Let $\mathscr S=(\mathcal S,\beta,d)$ be a Strategic Quotient. For every Harsanyi type space $(T_i,\pi_i)_i$ and associated maps $(\chi_i)_i$ satisfying Definition \ref{TS},
$$
d_i^m\bigl(\chi_i(t_i)\bigr)
=
\bigl(\ICR_i^0(t_i),\ldots,\ICR_i^m(t_i)\bigr)
$$
for every $t_i\in T_i$, $i\in I$, and $m\geq0$.}

\begin{proof}
We proceed by induction on $m$. At $m=0$, both sides equal $(A_i)$. Suppose the result holds at $m$. For $t_i\in T_i$, let
$$
p_i\coloneqq\pi_i(t_i)\circ(\id_K,\chi_{-i})^{-1}.
$$
The quotient property gives $\beta_i(p_i)=\chi_i(t_i)$. By the tower property and the induction hypothesis,
\begin{align*}
d_i^{m+1}\bigl(\chi_i(t_i)\bigr)
&=
\left(
d_i^m\bigl(\chi_i(t_i)\bigr),
\BR_i\bigl(p_i\circ(\id_K,b_{-i}^m)^{-1}\bigr)
\right)\\
&=
\left(
\ICR_i^0(t_i),\ldots,\ICR_i^m(t_i),
\BR_i\bigl(\pi_i(t_i)\circ(\id_K,\ICR_{-i}^m)^{-1}\bigr)
\right)\\
&=
\bigl(\ICR_i^0(t_i),\ldots,\ICR_i^{m+1}(t_i)\bigr).
\end{align*}
\end{proof}

\begin{claim}[Monotonicity of Hierarchies]\label{Monoghier}
For every player $i$, every $m\geq1$, every $s_i^m\in\mathcal S_i^m$, and every $0\leq l<m$,
$$
s_{i,l+1}\subseteq s_{i,l}.
$$
Consequently, every hierarchy in $\mathcal S_i$ is decreasing.
\end{claim}

\begin{proof}
We proceed by induction on $m$. At $m=1$, the assertion follows from $s_{i,0}=A_i$. Suppose the claim holds at order $m$ and fix $s_i^{m+1}\in\mathcal S_i^{m+1}$. The earlier inclusions are already contained in the prefix $s_i^m$, so only the inclusion between rounds $m$ and $m+1$ remains.

Choose $p_i^{m+1}\in\Delta_{K\times\mathcal S_{-i}^m}$ generating $s_i^{m+1}$. By the induction hypothesis, every opponent description in its support satisfies $s_{j,m}\subseteq s_{j,m-1}$. The induced distribution on $(k,s_{-i,m},s_{-i,m-1})$ is therefore a nested coupling of $\marg_{K,m}(p_i^{m+1})$ and $\marg_{K,m-1}(p_i^{m+1})$. Lemma \ref{MonoBR} gives
$$
\BR_i(\marg_{K,m}(p_i^{m+1}))
\subseteq
\BR_i(\marg_{K,m-1}(p_i^{m+1})).
$$
By definition, the two sides are $s_{i,m+1}$ and $s_{i,m}$.
\end{proof}

\paragraph{Lemma \ref{thm:ICR_hierarchies}}\emph{
\begin{itemize}
\item[(i)] A tuple $s^m\in\mathcal B^{m+1}$ belongs to $\mathcal S^m$ if and only if there exist a Harsanyi type space $(T,\pi)$ and a type profile $t\in T$ such that $s^m=(\ICR^\ell(t))_{\ell\leq m}$.
\item[(ii)] A sequence $s\in\mathcal B^\mathbb N$ belongs to $\mathcal S$ if and only if there exist a Harsanyi type space $(T,\pi)$ and a type profile $t\in T$ such that $s=(\ICR^\ell(t))_{\ell\geq0}$.
\end{itemize}
}

\begin{proof}
\smallskip\noindent\emph{Preliminary facts.}
Every hierarchy in $\mathcal S_i$ is decreasing by Claim \ref{Monoghier}. For every $p_i\in\Delta_{K\times\mathcal S_{-i}}$, the hierarchy $\beta_i(p_i)$ belongs to $\mathcal S_i$, because the projection of $p_i$ on $K\times\mathcal S_{-i}^{m-1}$ generates its first $m+1$ coordinates. Finally, every $\mathcal S_i^m$ is finite and nonempty. Nonemptiness follows inductively by choosing any belief on the finite nonempty set $K\times\mathcal S_{-i}^{m-1}$ and applying \eqref{Eq_BRSTS}.

\smallskip\noindent\emph{Part \emph{(i)}: finite hierarchies.}
Suppose first that $t$ is a type profile in a Harsanyi type space. Its order-zero description is $(A_i)_i\in\mathcal S^0$. If its order-$(m-1)$ description belongs to $\mathcal S^{m-1}$, the distribution induced by $\pi_i(t_i)$ on nature and opponents' order-$(m-1)$ descriptions satisfies \eqref{BRhierarchylevelm}. Induction gives that its order-$m$ description belongs to $\mathcal S^m$ for every $m$.

Conversely, we show by induction on $m$ that every $s^m\in\mathcal S^m$ is realized in a finite Harsanyi type space. The assertion is immediate at $m=0$. Fix $m\geq1$ and suppose it holds at order $m-1$. For every player $i$, choose a belief
$$
p_i^m\in\Delta_{K\times\mathcal S_{-i}^{m-1}}
$$
that generates $s_i^m$. Only finitely many opponent descriptions occur in these beliefs. Complete each of them to a profile in $\mathcal S^{m-1}$ and, by the induction hypothesis, choose a finite type space and a type profile realizing that completion.

Combine these finitely many spaces by taking disjoint copies of each player's types; old types retain their original beliefs inside their copy. Add one new type $t_i^*$ for each player. For every opponent description $r_{-i}^{m-1}$ used by $p_i^m$, choose a profile of opponent types in the combined space that realizes it. Give $t_i^*$ the probability that $p_i^m$ assigns to each state and corresponding opponent description. The induced distribution of the state and opponents' order-$(m-1)$ descriptions is therefore exactly $p_i^m$, including all its correlations. These beliefs can be specified separately for each player because no common prior is imposed. Equation \eqref{BRhierarchylevelm} now gives
$$
(\ICR_i^0(t_i^*),\ldots,\ICR_i^m(t_i^*))=s_i^m
$$
for every $i$, proving part \emph{(i)}.

In particular, every $s_i^m\in\mathcal S_i^m$ extends to an element of $\mathcal S_i$. Complete $s_i^m$ to a profile in $\mathcal S^m$ and realize that profile as above. The complete ICR hierarchy of the corresponding type extends $s_i^m$, and all its finite prefixes belong to the recursively constructed sets by the first direction of part \emph{(i)}.

\smallskip\noindent\emph{Maximal continuations.}
We compare continuations by how long their actions survive. A continuation is larger if it contains every action surviving in the smaller continuation at every round. For $s_i,\hat s_i\in\mathcal S_i$, write
$$
s_i\preceq_i\hat s_i
\quad\Longleftrightarrow\quad
s_{i,n}\subseteq\hat s_{i,n}
\text{ for every }n.
$$
For $s_i^m\in\mathcal S_i^m$, let
$$
\mathcal C_i(s_i^m)
\coloneqq
\{\hat s_i\in\mathcal S_i:\hat s_i^m=s_i^m\},
$$
denote its $\preceq_i$-maximal elements by
$\bar{\mathcal S}_i^m(s_i^m)$, and set
$$
\bar{\mathcal S}_i^m
\coloneqq
\bigcup_{s_i^m\in\mathcal S_i^m}
\bar{\mathcal S}_i^m(s_i^m).
$$

We need two properties: $\bar{\mathcal S}_i^m(s_i^m)$ is finite and nonempty, and every continuation in $\mathcal C_i(s_i^m)$ is dominated by one of its elements. Give $\mathcal B_i^{\mathbb N}$ the product topology. The set $\mathcal S_i$ is closed, since membership is determined by its finite truncations, and is therefore compact. The same holds for $\mathcal C_i(s_i^m)$.

A decreasing hierarchy is identified with its vector of deletion times
$$
\tau_a(s_i)\coloneqq
\inf\{n\in\mathbb N:a\notin s_{i,n}\},
\qquad a\in A_i,
$$
where $\inf\varnothing=\infty$. Under this identification,
$s_i\preceq_i\hat s_i$ if and only if
$\tau_a(s_i)\leq\tau_a(\hat s_i)$ for every $a$. To see that any set of deletion-time vectors has finitely many maximal elements, partition these elements according to the coordinates equal to $\infty$. There are finitely many such patterns. Within each pattern, the remaining coordinates lie in a finite power of $\mathbb N$, whose antichains are finite by Dickson's lemma. Distinct maximal elements are incomparable, so each pattern contains only finitely many of them.

Fix $s_i\in\mathcal C_i(s_i^m)$. On the compact set
$$
\{\hat s_i\in\mathcal C_i(s_i^m):s_i\preceq_i\hat s_i\},
$$
maximize
$$
\sum_{a\in A_i}\varphi(\tau_a),
\qquad
\varphi(n)=1-2^{-n},\quad \varphi(\infty)=1.
$$
Here $\mathbb N\cup\{\infty\}$ carries the topology in which $n\to\infty$. The deletion-time map is continuous because fixing a finite deletion time, or requiring it to exceed a given round, imposes only finitely many coordinate conditions. The displayed function is also continuous and strictly increases whenever some action survives longer and none survives for less time. A maximizer is therefore maximal in $\mathcal C_i(s_i^m)$ and dominates $s_i$. This proves both required properties. Since $\mathcal S_i^m$ is finite, $\bar{\mathcal S}_i^m$ is finite as well.

\smallskip\noindent\emph{Uniform stabilization.}
There exists $M\in\mathbb N$ such that, for every player $i$, every $m$, and every
$s_i\in\bar{\mathcal S}_i^m$, the hierarchy $s_i$ is constant from round $m+M$ onward. Since there are finitely many players, each $\bar{\mathcal S}_i^0$ is finite, and every decreasing hierarchy is eventually constant, a common $M$ exists at $m=0$. Suppose the assertion holds at $m-1$, and fix
$s_i\in\bar{\mathcal S}_i^m$. The idea is to replace each opponent's continuation by a maximal one agreeing through round $m-1$. By induction, these continuations stop changing by round $m-1+M$, so the player's best replies stop changing one round later.

Formally, let $N=m+M$. By the definition of $\mathcal S_i^N$, there is a belief on
$K\times\mathcal S_{-i}^{N-1}$ that generates $s_i^N$. Extend each opponent description in its finite support to a complete hierarchy. This gives a finitely supported belief
$p_i\in\Delta_{K\times\mathcal S_{-i}}$ satisfying
$$
d_i^N(\beta_i(p_i))=s_i^N.
$$

For every $(k,x_{-i})$ in the support of $p_i$ and every $j\neq i$, choose
$$
y_j\in\bar{\mathcal S}_j^{m-1}(x_j^{m-1})
\quad\text{with}\quad
x_j\preceq_j y_j.
$$
Replace each full support point jointly by
$(k,x_{-i})\mapsto(k,y_{-i})$, and denote the resulting belief by
$p_i^+$. The joint distribution of the old and new profiles preserves the state and is supported on nested action-set profiles at every round. Lemma \ref{MonoBR} therefore gives
$$
\beta_i(p_i)_n\subseteq\beta_i(p_i^+)_n
\quad\text{for every }n,
$$
with equality through round $m$ because every replacement preserves the order-$(m-1)$ opponent description. Since $\beta_i(p_i)$ agrees with $s_i$ through round $N$, it follows that $s_{i,n}\subseteq\beta_i(p_i^+)_n$ for every $n\leq N$.

By the induction hypothesis, every hierarchy in the support of $p_i^+$ is constant from round $m-1+M$ onward. Consequently,
$r_i\coloneqq\beta_i(p_i^+)$ is constant from round $m+M=N$ onward. For $n>N$, monotonicity of the hierarchy and constancy of $r_i$ give
$$
s_{i,n}\subseteq s_{i,N}\subseteq r_{i,N}=r_{i,n}.
$$
Thus $s_i\preceq_i r_i$. Moreover, $r_i\in\mathcal S_i$ and
$r_i^m=s_i^m$. Maximality of $s_i$ in $\mathcal C_i(s_i^m)$ implies
$r_i=s_i$. This completes the induction.

\smallskip\noindent\emph{One belief for every order.}
We can now prove that $\beta_i$ is onto. Fix $s_i\in\mathcal S_i$, and choose $m\geq1$ such that $s_i$ is constant from round $m$ onward. Let $N=m+M$ and extend a belief generating $s_i^N$ to a finitely supported belief
$p_i\in\Delta_{K\times\mathcal S_{-i}}$, and replace every hierarchy in its support by a dominating maximal continuation after round $m-1$, exactly as above. Let $p_i^+$ be the resulting belief and
$r_i=\beta_i(p_i^+)$. Then $r_i^m=s_i^m$, and
$s_{i,n}\subseteq r_{i,n}$ for every $n\leq N$. Since $s_i$ is constant after $m$ and $r_i$ is decreasing,
$$
r_{i,n}\subseteq r_{i,m}=s_{i,m}=s_{i,n}
\qquad\text{for every }n\geq m.
$$
Hence $r_i$ and $s_i$ agree through round $N$. The support of $p_i^+$ consists of maximal continuations after $m-1$, so $r_i$ is constant from round $m+M=N$ onward. It follows that $r_i=s_i$, proving that $\beta_i$ is onto.

\smallskip\noindent\emph{Part \emph{(ii)}: complete hierarchies.}
The direction from Harsanyi type spaces to $\mathcal S$ follows from the finite-order argument. For the converse, each $\mathcal S_i$ is countable, since a decreasing hierarchy is determined by the finitely many deletion times $(\tau_a)_{a\in A_i}$. For every $s_i\in\mathcal S_i$, choose
$$
p_i(s_i)\in\beta_i^{-1}(s_i).
$$
Take $T_i=\mathcal S_i$ and define $\pi_i(s_i)=p_i(s_i)$. Since $\mathcal S_i$ carries the full power set, these selections are measurable. An induction on $m$ gives
$$
\ICR_i^m(s_i)=s_{i,m}
$$
for every player, every $s_i\in\mathcal S_i$, and every $m$. Hence this single Harsanyi type space realizes every profile in $\mathcal S$.
\end{proof}

\paragraph{Lemma \ref{RepThm1a}} \emph{$(\mathcal S,\beta,d)$ is a Strategic Quotient.}

\begin{proof}
We verify in turn that $\beta$ is measurable and onto, that every Harsanyi type space maps to $(\mathcal S,\beta)$, and that $d$ is a strategic-description tower.

Lemma \ref{thm:ICR_hierarchies} establishes that each $\beta_i$ is onto. Each finite coordinate of $\beta_i$ is measurable by the argument of \cite[Lemma 1]{dekel2007interim}: taking marginals is measurable, and $\{p_i:\BR_i(p_i)=b_i\}$ is Borel on the finite-dimensional simplex for every $b_i\in\mathcal B_i$. Since $\mathcal S_i$ carries the $\sigma$-algebra generated by its coordinate projections, $\beta_i$ is measurable.

Let $(T,\pi)$ be a Harsanyi type space and define
$$
\chi_i(t_i)\coloneqq
\bigl(\ICR_i^m(t_i)\bigr)_{m\geq0}.
$$
Lemma \ref{thm:ICR_hierarchies} implies that $\chi_i(t_i)\in\mathcal S_i$, and \cite[Lemma 1]{dekel2007interim} implies that $\chi_i$ is measurable. If
$$
p_i\coloneqq
\pi_i(t_i)\circ(\id_K,\chi_{-i})^{-1},
$$
then the definitions of ICR and $\beta_i$ give
$$
d_i^m\bigl(\beta_i(p_i)\bigr)
=
d_i^m\bigl(\chi_i(t_i)\bigr)
\qquad\text{for every }m.
$$
Elements of $\mathcal S_i$ are complete coordinate sequences, so these equalities imply $\beta_i(p_i)=\chi_i(t_i)$. Hence the diagram in Definition \ref{TS} commutes and $(\mathcal S,\beta)$ is a type space quotient.

The maps $d_i^m$ are measurable coordinate projections. By the definition of $\beta_i$, for every $p_i\in\Delta_{K\times\mathcal S_{-i}}$,
$$
d_i^{m+1}\bigl(\beta_i(p_i)\bigr)
=
\left(
d_i^m\bigl(\beta_i(p_i)\bigr),
\BR_i\bigl(p_i\circ(\id_K,b_{-i}^m)^{-1}\bigr)
\right).
$$
Hence $d$ is a strategic-description tower.
\end{proof}

\paragraph{Theorem \ref{CoroEx}} \emph{The Strategic Quotient $(\mathcal S,\beta,d)$ is terminal. Consequently, the STS exists and is unique up to a unique isomorphism.}

\begin{proof}
Let $\widetilde{\mathscr S}=(\widetilde{\mathcal S},\widetilde\beta,\widetilde d)$ be a Strategic Quotient, and let $\widetilde b_i^m$ denote the last coordinate of $\widetilde d_i^m$. Define
$$
\zeta_i(\widetilde s_i)
\coloneqq
\bigl(\widetilde b_i^m(\widetilde s_i)\bigr)_{m\geq0}.
$$
Thus $\zeta_i$ sends a reduced type to its complete strategic-description tower. We show that this map is well defined and measurable, preserves beliefs and strategic descriptions, is onto, and is the only map with these properties.

We verify by induction on $m$ that every finite prefix belongs to $\mathcal S_i^m$. The assertion is immediate at order zero. For $m\geq1$, suppose it holds at order $m-1$ and fix $\widetilde s_i$. Because $\widetilde\beta_i$ is onto, choose $\widetilde p_i$ with $\widetilde\beta_i(\widetilde p_i)=\widetilde s_i$. By the induction hypothesis, $\widetilde d_{-i}^{m-1}$ takes values in $\mathcal S_{-i}^{m-1}$, so $(\id_K,\widetilde d_{-i}^{m-1})$ induces from $\widetilde p_i$ a belief on $K\times\mathcal S_{-i}^{m-1}$. The tower equations at orders $0,\ldots,m-1$ give the coordinates of $\widetilde d_i^m(\widetilde s_i)$. Thus $\widetilde d_i^m(\widetilde s_i)\in\mathcal S_i^m$, and $\zeta_i(\widetilde s_i)\in\mathcal S_i$.

The map $\zeta_i$ is measurable because each coordinate is measurable and the $\sigma$-algebra on $\mathcal S_i$ is generated by the coordinate projections.

By construction,
$$
d_i^m\circ\zeta_i=\widetilde d_i^m.
$$
For $\widetilde p_i\in\Delta_{K\times\widetilde{\mathcal S}_{-i}}$, set
$$
p_i\coloneqq
\widetilde p_i\circ(\id_K,\zeta_{-i})^{-1}.
$$
The order-zero coordinates of $\zeta_i(\widetilde\beta_i(\widetilde p_i))$ and $\beta_i(p_i)$ both equal $A_i$. At every order $m+1$, the tower property and $\widetilde b_{-i}^m=b_{-i}^m\circ\zeta_{-i}$ give
\begin{align*}
\widetilde b_i^{m+1}\bigl(\widetilde\beta_i(\widetilde p_i)\bigr)
&=
\BR_i\bigl(
\widetilde p_i\circ(\id_K,\widetilde b_{-i}^m)^{-1}
\bigr)\\
&=
\BR_i\bigl(
p_i\circ(\id_K,b_{-i}^m)^{-1}
\bigr)
=
b_i^{m+1}\bigl(\beta_i(p_i)\bigr).
\end{align*}
Thus the quotient diagram in Definition \ref{reduction} commutes.

It remains to show that $\zeta_i$ is onto. Fix $s_i\in\mathcal S_i$. By Lemma \ref{thm:ICR_hierarchies}, some type $t_i$ in some Harsanyi type space has ICR hierarchy $s_i$. Let $\widetilde\chi$ be the maps from that type space to $\widetilde{\mathcal S}$ supplied by Definition \ref{TS}. Lemma \ref{StratThm1} gives
$$
\widetilde d_i^m\bigl(\widetilde\chi_i(t_i)\bigr)=s_i^m
\qquad\text{for every }m,
$$
so $\zeta_i(\widetilde\chi_i(t_i))=s_i$. Hence $\zeta$ is a morphism.

Finally, if $\widehat\zeta$ is another morphism from $\widetilde{\mathscr S}$ to $(\mathcal S,\beta,d)$, then
$$
d_i^m\circ\widehat\zeta_i
=
\widetilde d_i^m
=
d_i^m\circ\zeta_i
\qquad\text{for every }m.
$$
Elements of $\mathcal S_i$ are complete sequences of finite-order descriptions, so $\widehat\zeta_i=\zeta_i$. Therefore $(\mathcal S,\beta,d)$ is terminal. The uniqueness conclusion follows from the uniqueness of terminal objects.
\end{proof}

\subsubsection{Finiteness of the Continuation State Space}\label{FiniteSTSProof}

For $s_i\in\mathcal S_i$, let
$$
s_{i,\infty}\coloneqq\bigcap_{n\in\naturals}s_{i,n}.
$$
By Claim \ref{Monoghier}, the sequence $(s_{i,n})_n$ is decreasing. Since $A_i$ is finite, it is eventually constant at $s_{i,\infty}$.
Let $\proj_\infty(s_i)=s_{i,\infty}$ and let $\proj_{-i,\infty}$ denote the corresponding product map.

\begin{lemma}[Limit best replies]\label{LimitBR}
Let $p_i\in\Delta_{K\times\mathcal S_{-i}}$. For $n\in\naturals$, let
$$
p_i^n
\coloneqq
p_i\circ(\id_K,\proj_{-i,n})^{-1},
\qquad
p_i^\infty
\coloneqq
p_i\circ(\id_K,\proj_{-i,\infty})^{-1}.
$$
Then
$$
\beta_i(p_i)_\infty=\BR_i(p_i^\infty).
$$
\end{lemma}

\begin{proof}
We first use closedness to show that every action surviving at all sufficiently large finite rounds remains a best reply under the limit marginal. Monotonicity gives the reverse inclusion.

The distributions $p_i^n$ and $p_i^\infty$ are obtained from the same draw under $p_i$. Hence
$$
\lVert p_i^n-p_i^\infty\rVert_{\mathrm{TV}}
\leq
p_i\{s_{-i,n}\neq s_{-i,\infty}\}.
$$
Every hierarchy is eventually constant, so the indicator of the event on the right converges pointwise to zero. Dominated convergence gives $p_i^n\to p_i^\infty$ in total variation.

For every $a_i\in A_i$, the set
$$
\{p\in\Delta_{K\times\mathcal B_{-i}}:a_i\in\BR_i(p)\}
$$
is closed. Indeed, it is the projection on $p$ of the compact set of pairs $(p,q)$ satisfying $q\in\mathcal Q_i(p)$ and
$$
\mathbb E_q[u_i(a_i,a_{-i},k)]
\geq
\mathbb E_q[u_i(a_i',a_{-i},k)]
\quad\text{for every }a_i'\in A_i.
$$
The projection is compact and hence closed.

Let $b_i=\beta_i(p_i)_\infty$. For all sufficiently large $n$,
$\BR_i(p_i^n)=\beta_i(p_i)_{n+1}=b_i$, because the hierarchy $\beta_i(p_i)$ is eventually constant. Closedness therefore gives
$b_i\subseteq\BR_i(p_i^\infty)$. Conversely, the same draw from $p_i$ couples $p_i^\infty$ and $p_i^n$ and satisfies
$s_{j,\infty}\subseteq s_{j,n}$ for every $j\neq i$. Lemma \ref{MonoBR} gives
$$
\BR_i(p_i^\infty)\subseteq\BR_i(p_i^n)=b_i
$$
for all sufficiently large $n$.
\end{proof}

We now prove finiteness by placing the continuation problem in finite domains. This reduction uses the already constructed hierarchy space, including the base sets $R^{0,\infty}$ below; it does not assert a practical algorithm from the payoff matrix. The sets $R^{q,\infty}$ record the action sets that can occur jointly at selected rounds and in the limit coordinate. The map $\Phi_r$ advances every selected round by one, while $\iota_r$ inserts round zero. Two round counts will be treated as equivalent when they have the same effect under every relevant update. Finally, the first and last rounds of the constant blocks provide a finite summary of any prefix.

We first record several coordinates of a hierarchy simultaneously. Repeated rounds are allowed. For $q\geq0$,
$\mathbf n=(n_1,\ldots,n_q)$ with
$0\leq n_1\leq\cdots\leq n_q$, and
$x_i\in\mathcal B_i^{\mathbb N}$, define
$$
\proj_i^{\mathbf n,\infty}(x_i)
\coloneqq
(x_{i,n_1},\ldots,x_{i,n_q},x_{i,\infty}),
\qquad
\proj_{-i}^{\mathbf n,\infty}
\coloneqq
(\proj_j^{\mathbf n,\infty})_{j\neq i}.
$$
Then
$$
R_i^{q,\infty}(\mathbf n)
\coloneqq
\proj_i^{\mathbf n,\infty}(\mathcal S_i).
$$
For $q=0$, write
$$
R_i^{0,\infty}
\coloneqq
\{s_{i,\infty}:s_i\in\mathcal S_i\}.
$$
Thus $R_i^{q,\infty}$ has $q+1$ coordinates, including the limit coordinate. Let
$R^{q,\infty}(\mathbf n)
=(R_i^{q,\infty}(\mathbf n))_{i\in I}$.

For $r\geq1$, let
$$
\mathfrak R^r
\coloneqq
\prod_{i\in I}
\left(2^{\mathcal B_i^r}\setminus\{\varnothing\}\right).
$$
Fix $D=(D_i)_i\in\mathfrak R^r$.
For $1\leq\ell\leq r$, the product projection is
$\proj_{-i,\ell}(x_{-i})=(x_j^\ell)_{j\neq i}$.
For $\mu_i\in\Delta_{K\times D_{-i}}$, let
$$
\mu_i^\ell
\coloneqq
\mu_i\circ
(\id_K,\proj_{-i,\ell})^{-1}.
$$
Define $\Phi_r:\mathfrak R^r\to\mathfrak R^r$ by
\begin{equation}\label{CoordinateOperator}
\Phi_r(D)_i
\coloneqq
\left\{
\bigl(\BR_i(\mu_i^1),\ldots,\BR_i(\mu_i^r)\bigr):
\mu_i\in\Delta_{K\times D_{-i}}
\right\}.
\end{equation}
For $r\geq2$, define
$\iota_r:\mathfrak R^{r-1}\to\mathfrak R^r$ by
\begin{equation}\label{ZeroInsertion}
(\iota_rD)_i
\coloneqq
\left\{
(A_i,x_i^1,\ldots,x_i^{r-1}):
(x_i^1,\ldots,x_i^{r-1})\in D_i
\right\}.
\end{equation}

\begin{lemma}[Coordinate recursion]\label{CoordinateRecursion}
For every $q\geq1$ and
$0\leq n_1\leq\cdots\leq n_q$,
\begin{equation}\label{CoordinateShift}
R^{q,\infty}(n_1+1,\ldots,n_q+1)
=
\Phi_{q+1}
\bigl(R^{q,\infty}(n_1,\ldots,n_q)\bigr).
\end{equation}
Moreover,
\begin{equation}\label{CoordinateZero}
R^{q,\infty}(0,n_2,\ldots,n_q)
=
\iota_{q+1}
\bigl(R^{q-1,\infty}(n_2,\ldots,n_q)\bigr).
\end{equation}
The empty argument on the right is omitted when $q=1$.
\end{lemma}

\begin{proof}
Let $\mathbf n=(n_1,\ldots,n_q)$. We prove the two inclusions in \eqref{CoordinateShift}. Fix a player $i$. For the inclusion from left to right, take $s_i\in\mathcal S_i$ and choose
$p_i\in\Delta_{K\times\mathcal S_{-i}}$ such that
$\beta_i(p_i)=s_i$. Let
$$
\mu_i
\coloneqq
p_i\circ(\id_K,\proj_{-i}^{\mathbf n,\infty})^{-1}.
$$
This distribution is supported on
$K\times R_{-i}^{q,\infty}(n_1,\ldots,n_q)$.
Its first $q$ action-set marginals generate
$(s_{i,n_1+1},\ldots,s_{i,n_q+1})$, and its limit marginal generates $s_{i,\infty}$ by Lemma \ref{LimitBR}. This proves the first inclusion.

For the reverse inclusion, take
$y_i\in\Phi_{q+1}(R^{q,\infty}(\mathbf n))_i$. By the definition of $\Phi_{q+1}$, some distribution
$$
\mu_i\in\Delta_{K\times R_{-i}^{q,\infty}(\mathbf n)}
$$
generates the coordinates of $y_i$.
For every player $j\neq i$ and every tuple
$x_j\in R_j^{q,\infty}(\mathbf n)$, choose
$\rho_j(x_j)\in\mathcal S_j$ such that
$\proj_j^{\mathbf n,\infty}(\rho_j(x_j))=x_j$. Let
$\rho_{-i}=(\rho_j)_{j\neq i}$ and define
$$
p_i
\coloneqq
\mu_i\circ(\id_K,\rho_{-i})^{-1}.
$$
Since each $\rho_j$ is a section of the coordinate projection,
$$
p_i\circ(\id_K,\proj_{-i}^{\mathbf n,\infty})^{-1}
=
\mu_i.
$$
Thus
$(\beta_i(p_i)_{n_1+1},\ldots,\beta_i(p_i)_{n_q+1})$ gives the first $q$ components of $y_i$, and Lemma \ref{LimitBR} gives $\beta_i(p_i)_\infty$, its final component. Hence
$y_i\in R_i^{q,\infty}(n_1+1,\ldots,n_q+1)$, proving the reverse inclusion. The construction is performed separately for each player; no common-prior condition links the resulting beliefs.

Equation \eqref{CoordinateZero} follows from $s_{i,0}=A_i$.
\end{proof}

\begin{lemma}[Gap representation]\label{GapRepresentation}
Let $q\geq1$ and
$0\leq n_1\leq\cdots\leq n_q$. Set
$$
g_1=n_1,
\qquad
g_\ell=n_\ell-n_{\ell-1}
\quad\text{for }2\leq\ell\leq q.
$$
Then
\begin{equation}\label{GapRepresentationFormula}
R^{q,\infty}(n_1,\ldots,n_q)
=
\left(
\Phi_{q+1}^{g_1}
\circ\iota_{q+1}
\circ\Phi_q^{g_2}
\circ\iota_q
\circ\cdots\circ
\Phi_2^{g_q}
\circ\iota_2
\right)
(R^{0,\infty}).
\end{equation}
\end{lemma}

\begin{proof}
Read the composition from right to left. The first insertion and shift give
$$
(\Phi_2^{g_q}\circ\iota_2)(R^{0,\infty})
=
R^{1,\infty}(g_q).
$$
After the next insertion and shift, the two selected rounds are
$g_{q-1}$ and $g_{q-1}+g_q$. Continuing in this way, after processing
$g_h,\ldots,g_q$ the selected rounds are
$$
g_h,\quad g_h+g_{h+1},\quad\ldots,\quad g_h+\cdots+g_q.
$$
The final selected rounds are therefore
$$
(g_1,g_1+g_2,\ldots,g_1+\cdots+g_q)
=
(n_1,\ldots,n_q).
$$
Zero gaps allow repeated rounds.
\end{proof}

Set
$$
Q\coloneqq2\max_{i\in I}|A_i|+1.
$$
For $n,n'\in\naturals$, write
\begin{equation}\label{PowerEquivalence}
n\equiv_Q n'
\quad\Longleftrightarrow\quad
\Phi_r^n=\Phi_r^{n'}
\text{ as maps on }\mathfrak R^r
\text{ for every }1\leq r\leq Q.
\end{equation}
Let $[n]_Q$ denote the equivalence class of $n$.

\begin{lemma}[Finite power signatures]\label{FiniteSignatures}
The relation $\equiv_Q$ has finitely many equivalence classes. Moreover,
$$
n\equiv_Q n'
\quad\Longrightarrow\quad
n+h\equiv_Q n'+h
\quad\text{for every }h\in\naturals.
$$
\end{lemma}

\begin{proof}
For every $r\leq Q$, the set $\mathfrak R^r$ is finite. Hence the powers of $\Phi_r$ take only finitely many values as maps on $\mathfrak R^r$. This proves the first assertion. For the second,
$$
\Phi_r^{n+h}
=
\Phi_r^h\circ\Phi_r^n
=
\Phi_r^h\circ\Phi_r^{n'}
=
\Phi_r^{n'+h}.
$$
\end{proof}

For every weakly increasing vector
$\mathbf n=(n_1,\ldots,n_q)$, define its \emph{gap signature} by
$$
\operatorname{gap}_Q(\mathbf n)
\coloneqq
\left([n_1]_Q,([n_\ell-n_{\ell-1}]_Q)_{\ell=2}^q\right).
$$
If $q+1\leq Q$ and
$\operatorname{gap}_Q(\mathbf n)
=\operatorname{gap}_Q(\hat{\mathbf n})$, Lemma \ref{GapRepresentation} and the definition of $\equiv_Q$ imply
\begin{equation}\label{EqualGapCoordinates}
R^{q,\infty}(\mathbf n)
=
R^{q,\infty}(\hat{\mathbf n}).
\end{equation}

We now identify the finite coordinates that determine a decreasing hierarchy. Let
$x_i=(x_{i,n})_{n\geq0}\in\mathcal B_i^\mathbb N$ be decreasing, and write $x_{i,\infty}\coloneqq\bigcap_nx_{i,n}$. Suppose its strict transitions occur at rounds
$$
1\leq b_1<\cdots<b_d,
$$
and set $b_0=0$ and
$e_h=b_{h+1}-1$ for $0\leq h<d$.
Thus $x_i$ is constant on each interval $[b_h,e_h]$ for $h<d$ and at every round $n\geq b_d$. Define
$$
\mathbf c(x_i)
\coloneqq
(b_0,e_0,b_1,e_1,\ldots,b_{d-1},e_{d-1},b_d),
$$
with the convention $\mathbf c(x_i)=(b_0)$ when $d=0$.

\begin{lemma}[Critical coordinates]\label{CriticalCoordinates}
Let $c_1,\ldots,c_{2d+1}$ be the entries of
$\mathbf c(x_i)$. Then
$$
x_i\in\mathcal S_i
$$
if and only if
$$
(x_{i,c_1},\ldots,x_{i,c_{2d+1}},x_{i,\infty})
\in
R_i^{2d+1,\infty}
(c_1,\ldots,c_{2d+1}).
$$
\end{lemma}

\begin{proof}
The forward implication follows from the definition of $R_i^{2d+1,\infty}$. Conversely, suppose the displayed tuple is generated by some $s_i\in\mathcal S_i$. Both $s_i$ and $x_i$ are decreasing. On every finite constant interval,
$$
s_{i,b_h}=x_{i,b_h}=x_{i,e_h}=s_{i,e_h},
$$
and both sequences decrease, so every coordinate between these equal endpoints has their common value. Thus
$s_{i,n}=x_{i,n}$ for every $b_h\leq n\leq e_h$.
Similarly,
$s_{i,b_d}=x_{i,b_d}=x_{i,\infty}=s_{i,\infty}$ gives equality at every $n\geq b_d$. Hence $s_i=x_i$.
\end{proof}

Consider a prefix
$s_i^m=(s_{i,0},\ldots,s_{i,m})\in\mathcal S_i^m$.
Let
$$
0=b_0<b_1<\cdots<b_L\leq m
$$
be the first rounds of its successive constant blocks, and set
$e_h=b_{h+1}-1$ for $0\leq h<L$.
Define its \emph{prefix signature} by
\begin{equation}\label{PrefixCode}
\operatorname{sig}_i(s_i^m)
\coloneqq
\left(
(s_{i,b_h})_{h=0}^L,\,
([e_h-b_h]_Q)_{h=0}^{L-1},\,
[m-b_L]_Q
\right).
\end{equation}

\begin{lemma}[Prefix signatures]\label{PrefixSignature}
If
$$
\operatorname{sig}_i(s_i^m)
=
\operatorname{sig}_i(\hat s_i^{\hat m}),
$$
then
$$
\mathcal T_i(s_i^m)
=
\mathcal T_i(\hat s_i^{\hat m}).
$$
\end{lemma}

\begin{proof}
The equality of signatures implies that the two prefixes visit the same chain of action sets. Let
$$
0=b_0<\cdots<b_L\leq m,
\qquad
0=\hat b_0<\cdots<\hat b_L\leq\hat m
$$
be the first rounds of their corresponding constant blocks, and set
$e_h=b_{h+1}-1$ and $\hat e_h=\hat b_{h+1}-1$ for $h<L$.

It is enough to show that every candidate tail extends both prefixes or neither. Let
$\tau_i=(\tau_{i,0},\tau_{i,1},\ldots)\in\mathcal B_i^\mathbb N$ be any candidate continuation with
$\tau_{i,0}=s_{i,m}=\hat s_{i,\hat m}$, and form
$$
x_i
=
(s_{i,0},\ldots,s_{i,m},\tau_{i,1},\tau_{i,2},\ldots),
$$
$$
\hat x_i
=
(\hat s_{i,0},\ldots,\hat s_{i,\hat m},
\tau_{i,1},\tau_{i,2},\ldots).
$$
The sequence $x_i$ is decreasing if and only if $\hat x_i$ is decreasing. If they are not decreasing, neither belongs to $\mathcal S_i$.

Suppose they are decreasing. Since the prefixes visit the same chain of action sets and the candidate tail is the same, $x_i$ and $\hat x_i$ have the same number $d$ of strict transitions.

For a decreasing sequence $z_i$ and a cut round $c$, let
$\mathbf c_c(z_i)$ be obtained from $\mathbf c(z_i)$ by inserting $c$ immediately after the first round of the constant block containing it. Repeated rounds are retained when $c$ is already the first or last round of that block. Set
$$
\mathbf a=\mathbf c_m(x_i),
\qquad
\hat{\mathbf a}=\mathbf c_{\hat m}(\hat x_i).
$$
These vectors have the same length, and
\begin{equation}\label{EqualCutCoordinates}
\proj_i^{\mathbf a,\infty}(x_i)
=
\proj_i^{\hat{\mathbf a},\infty}(\hat x_i),
\end{equation}
because the two prefixes visit the same action sets and the continuation after their respective cuts is the same.

We next compare their gaps. Before the cut, each gap is the count
$e_h-b_h$ across a completed constant block, the one-round gap
$b_{h+1}-e_h=1$ between consecutive blocks, or the count
$m-b_L$ from the start of the current block to the cut. The corresponding gaps for the second prefix are
$\hat e_h-\hat b_h$, $1$, and $\hat m-\hat b_L$. Equation \eqref{PrefixCode} assigns the same $\equiv_Q$ class to the first and third kinds, while the second kind is identical. Every gap after the cut is identical because the candidate tail is the same. Hence
\begin{equation}\label{EqualCutGaps}
\operatorname{gap}_Q(\mathbf a)
=
\operatorname{gap}_Q(\hat{\mathbf a}).
\end{equation}

The augmented vectors have at most $2d+2$ finite coordinates. Since $d\leq|A_i|-1$, their number of coordinates including the limit coordinate is at most
$$
2d+3
\leq
2|A_i|+1
\leq Q.
$$
Equations \eqref{EqualGapCoordinates} and \eqref{EqualCutGaps} therefore give
$$
R^{q,\infty}(\mathbf a)
=
R^{q,\infty}(\hat{\mathbf a}),
$$
where $q$ is the common number of finite coordinates.

Adding the cut round does not change the criterion in Lemma \ref{CriticalCoordinates}. If a hierarchy generates the augmented tuple, deleting the cut coordinate gives the critical tuple, and the lemma implies that the hierarchy coincides with the candidate at every round. Combining this fact with \eqref{EqualCutCoordinates} and the equality of the two coordinate sets gives
$$
x_i\in\mathcal S_i
\quad\Longleftrightarrow\quad
\hat x_i\in\mathcal S_i.
$$
The candidate tail belongs to both continuation sets or to neither.
\end{proof}

\paragraph{Proof of Theorem \ref{FiniteO}.}
A strict chain of nonempty subsets of $A_i$ has at most $|A_i|-1$ strict transitions. There are finitely many such chains. The relation $\equiv_Q$ has finitely many classes, and every prefix signature contains a bounded number of these classes. Hence there are finitely many prefix signatures. Lemma \ref{PrefixSignature} shows that the continuation set of a prefix depends only on its signature. Therefore $\Omega_i$ is finite.

\end{document}